\documentclass[11pt,a4paper]{amsart}

\usepackage[T1]{fontenc}
\usepackage[utf8]{inputenc}
\usepackage{lmodern}
\usepackage{amsmath,amssymb,amsthm}
\usepackage{geometry}
\usepackage{hyperref}
\usepackage{enumitem}
\usepackage{booktabs}
\usepackage{graphicx}

\newtheorem{theorem}{Theorem}
\newtheorem{proposition}[theorem]{Proposition}

\newtheorem{conjecture}[theorem]{Conjecture}

\newtheorem{fact}[theorem]{Fact}
\newtheorem{definition}[theorem]{Definition}
\newtheorem{algorithm}[theorem]{Algorithm}

\newcommand{\sgn}{\operatorname{sgn}}

\title{Root Isolation for Analytic Functions\\
Using Cubic Hermite Interpolation}
\author{Christophe Raffalli}
\address{G.A.A.T.I. \\ Université de la Polynésie Française}
\email{christophe@raffalli.eu}
\date{\today}

\begin{document}

\begin{abstract}
We present a root isolation algorithm for an
analytic function. The method combines interval subdivision with a
heuristic test based on local cubic Hermite interpolation. Whenever the
function is sufficiently close to its cubic Hermite interpolant at a few
sample points, we assume that the
critical points of the cubic interpolant give a valid root isolation for the
original function.

Such an algorithm cannot, in general, provide a rigorous certification of
the computed intervals. Nevertheless, extensive numerical experiments show
that it is highly efficient and almost always produces the correct results.
We hope that a probabilistic or exact algorithm can be derived from this work.
\end{abstract}

\maketitle

\section{Introduction}

In this paper, we propose a root isolation algorithm based on a local
interpolation criterion. The key idea is to replace symbolic
root-counting tests by a numerical comparison between the function and a local
Hermite interpolant.

The isolation of the real roots of a polynomial, or more generally of an
analytic function, is a fundamental problem in
computer algebra and computational geometry. Given a function
$f:\mathbb{R}\to\mathbb{R}$, the objective is to compute a finite
collection of pairwise disjoint intervals, each containing exactly one real
root of $f$ and such that every real root belongs to one of these intervals.

Several methods have been developed for isolating the real roots of a univariate
polynomial. Among the most classical ones are:

\begin{itemize}
  \item \emph{Continued fraction methods} \cite{TSIGARIDAS2008158}: this method
    transforms the polynomial using Möbius substitutions and recursively
    computes continued fraction expansions to isolate roots, without
    explicitly subdividing the real line.

\item \emph{Descartes' method} \cite{ROUILLIER200433, KAR2016, SAGRALOFF201646}:
  based on Descartes' rule of signs, it uses
  the number of sign variations of transformed polynomials to bound the number
  of roots in an interval. It recursively subdivides intervals until each
  interval contains at most one root.

  \item \emph{Sturm's method}: it uses Sturm sequences to count exactly the
  number of distinct real roots in any interval. Combined with recursive
  subdivision, it provides a certified root isolation algorithm.
\end{itemize}

There is less work on root isolation for general (non-polynomial)
functions. We may cite the work of John P. Boyd \cite{BOYD2013} and
the paper of Cleonice F. Bracciali and Michael Carley
\cite{BC2017}. Similarly to our approach, they use a local interpolation: Chebyshev for
the former, Taylor expansion for the latter. The main difference with our work
is that we use a fixed degree 3 Hermite interpolation.

The decision tests used during root isolation generally have a quadratic
complexity with respect to the degree, typically $O(d^2)$. This complexity comes either from
the computation of Taylor shifts, as in Descartes-based methods, or from the
basic construction and evaluation of Sturm sequences. In the case of
Descartes' rule, we
can abort the computation as soon as we detect too many sign changes to
mitigate the cost.
Although asymptotically
faster algorithms based on FFT techniques can reduce this complexity for very
large degrees, they are usually not advantageous in practice for the moderate
degrees encountered in most applications.

Our interpolation-based test has a complexity proportional to the cost of
a few evaluation of the function and its derivative. For a polynomial of degree $d$, a direct evaluation
generally costs $O(d)$ operations, while for special families such as Chebyshev
or Mignotte polynomials, this cost can be reduced to $O(1)$.
This leads to a root isolation algorithm that appears linear in our
experiments for those classes of polynomials.

Another important advantage of our approach is that it does not
require to expand the polynomial and the test can be applied without
modification to analytic functions.

Finally, unlike certified symbolic approaches, our method is designed as a
numerical algorithm relying on standard floating-point arithmetic, using
either 64-bit double-precision or even 32-bit single-precision floating point
arithmetic (the latter is standard on GPUs). Despite this limited precision, the
method exhibits very good numerical stability in practice. This numerical
stability is also improved by the fact that we do not need to expand the
polynomials.

For instance, at the end of this paper, we report an experiment to visualize
surfaces using GPUs with only 32-bit floating point arithmetic.  It is very
important to be able to evaluate the polynomial and its derivative restricted
to a line, by computing the polynomial and its gradient in $\mathbb{R}^3$.

We are able to prove that our algorithm indeed terminates and produces intervals
that contain at least one root of the provided analytic function. The
algorithm is parameterized by two constants $\varepsilon \in ]0,1[$ and $n \in \mathbb{N}^*$.
The integer $n$ is the number of point where we test the quality of the Hermite
interpolation, by checking that a quantity $Q(x) < \varepsilon$.
Our experiments indicate that very reasonable values of these constants, $n =
3$ or $6$ for instance, are sufficient for many families of polynomials or
analytic functions.

We hope that an exact or a probabilistic algorithm may be within reach by
developing the idea in this work.

\section{Description of the Algorithm}

\subsection*{Hermite interpolation predicate}

Let $f$ be an analytic function on an open set $\Omega \subset \mathbb{R}$ and let $h$ be the cubic
Hermite interpolant of $f$ on an interval $I=[a,b] \subset \Omega$, constructed from the
endpoint values and derivatives of $f$.
We obtain, when $a \neq b$:
\[
h(x)=c_0+c_1u+c_2u^2+c_3u^3,\qquad \text{where} u=x-a ,
\]
where
\[
\begin{aligned}
c_0 &= f(a),\\[1mm]
c_1 &= f'(a),\\[1mm]
c_2 &= \frac{3(f(b)-f(a))}{(b-a)^2}
      -\frac{2f'(a)+f'(b)}{b-a},\\[1mm]
c_3 &= \frac{2(f(a)-f(b))}{(b-a)^3}
      +\frac{f'(a)+f'(b)}{(b-a)^2}.
\end{aligned}
\]

Remarks:
\begin{itemize}
  \item For numerical stability, we choose $u = x - a$ if $|f(a)| \leqslant |f(b)|$
    otherwise we choose $u = x - b$ (exchanging $a$ and $b$ in the formula for
    the coefficients).
  \item Expressing the interpolant as an homogeneous polynomial in both
    $(x-a)$ and $(x-b)$ yields a simpler formula but is less stable
    numerically.
\end{itemize}

\begin{fact}\label{fact:hermite}
  Let $f$ be of class $C^4(I)$ where $I = [A,B] \subset \mathbb{R}$.
  Let $h$ be the cubic Hermite interpolant of $f$ in $a,b \in I$,
  Set $E(x)=f(x)-h(x)$, $L = b - a$ and $M = \|f^{(4)}\|_\infty =
  \max_{x \in I} |f^{(4)}(x)|$.
  \[\begin{array}{rcccl}
    |E(x)| &=& \frac{|f^{(4)}(\chi_x)|}{4!}(x-a)^2(x-b)^2 &\leqslant& \frac{M}{384}L^4 \cr
    |E'(x)| &=& \frac{|f^{(4)}(\xi_x)|}{3!}(x-a)(x-b)(x-m_x) &\leqslant& \frac{2M}{81}L^3 \cr
    |E''(x)| &=& \frac{|f^{(4)}(\xi_x)|}{2!}(x-k_x)(x-l_x) &\leqslant& \frac{M}{2}L^2\cr
  \end{array}\]
\end{fact}

\begin{proof}
  The first equality is the classical result for Hermite interpolation.
  For the second, by Rolle's theorem, we get a point $m_x$ in $]a,b[$ such that $E'(m_x) =
      0$ and the Lagrange interpolation error gives the equality and
      the fact that $\max(\{(x-a)(x-b)|x - m|, (x,m) \in
    ]a,b[^2\}) = \frac{4L^3}{27}$ gives the inequality.
      For the third one, from $E'(a) = E'(m_x) = E'(b) = 0$, we get
      two points with $a < k_x < l_x < b$ with $E''(k_x) = E''(l_x) = 0$ hence the
      result, again using Lagrange interpolation.
\end{proof}

\begin{fact}\label{fact:isotopic}
  Let $f$ be a continuously differentiable function on $I=[a,b]$ with no multiple roots
  in $I$ and let $h$ be the cubic
  Hermite interpolant of $f$ on $I$.
  If the set $V_f = \{ x \in I, f(x) = 0\}$ is not isotopic to the set
  $V_h = \{ x \in I, h(x) = 0\}$, then there exists one point $x_0 \in I$ and $\lambda \leqslant 0$ such that:
  $$h(x_0) = \lambda f(x_0) \hbox{ and } h'(x_0) = \lambda f'(x_0)$$
\end{fact}

Remark: this result is stronger than just saying that the Wronskian $f h' - h
f' = 0$ at some point.

\begin{proof}
  We consider $h_t(x) = t h(x) + (1 - t) f(x)$, If $V_f$ and $V_h$ are not
  isotopic, then there exists $t \in [0, 1]$ such that
  $h_t(x) = h'_t(x) = 0$. This implies that $t h(x) = (t - 1) f(x)$ and
  $t h'(x) = (t - 1) f'(x)$. As $f$ has no multiple roots in $I$, we have $t >
  0$, hence we have the desired result with $\lambda = \frac{t-1}{t}$.
\end{proof}

This fact suggests the following definition for the \emph{quality} of the
interpolant at $x$. This is justified by the following fact.

\begin{definition}[quality of the interpolant]\label{def:quality}
For $f$ analytic and non identically zero on $\Omega \supset I$ we define a notion of quality of the interpolant at
$x$, expressing that we are far from a point satisfying the conclusion of the
above fact:

\[Q(x) = \frac{|f(x)-h(x)|}{|f(x)|+|h(x)|}
\frac{|f'(x)-h'(x)|}{|f'(x)|+|h'(x)|}
\]

From the analyticity of $f$, we get that $Q$ can be extended by continuity
when $f(x_0) = h(x_0) = 0$ or $f'(x_0) = h'(x_0) = 0$. Indeed, assume that $f(x_0) =
h(x_0) = 0$. If $h = 0$ on $\Omega$, then the corresponding factor is $1$ outside $x_0$ and extends by
continuity to $1$. Otherwise, we have $f(x_0+t) = \alpha t^p
+ o(t^p)$ ($\alpha \neq 0, 1 \leqslant p$) and $h(x_0+t) = \beta t^q + o(t^q)$ ($\beta \neq 0, 1 \leqslant q \leqslant
3$). In this case we have
\[ \lim_{t \to 0} \frac{|f(x_0+t)-h(x_0+t)|}{|f(x_0+t)|+|h(x_0+t)|} = \begin{cases}
  1 \hbox{ if } p \neq q \\
  \frac{|\alpha - \beta|}{|\alpha|+|\beta|} \hbox{ if } p = q
\end{cases}\]
The same holds for $f'$ and $h'$.
\end{definition}

\begin{fact}
  For all $x \in I$, $Q(x) \in [0,1]$. Moreover, $Q(x) = 1$ iff both $\sgn(f(x)h(x)) \leqslant 0 $ and
  $\sgn(f'(x)h'(x)) \leqslant 0 $, and in particular for any point $x$ satisfying the
  conclusion of fact \ref{fact:isotopic}.
\end{fact}

\begin{proof}
  Immediate.
\end{proof}

\begin{fact}[geometric definition for $Q$]
  Let $\theta(x,y)=\angle\bigl((x,y),(1,1)\bigr)$ denote the angle of the
  vectors $(x,y)$ and $(1,1)$.
  Then, we have \[Q(x) = \min\left(1, |\tan(\theta(f(x),h(x)))|\right)\min\left(1,
  |\tan(\theta(f'(x),h'(x)))|\right).\]
\end{fact}

\begin{proof}
  For any $(x,y)\neq(0,0)$, \[ \tan(\theta(x,y))=\frac{x-y}{x+y}. \]
  The result follows distinguishing the cases where $x$ and $y$ have the
  same sign or opposite sign.
\end{proof}

The two next propositions are properties of $Q$ needed to prove the
termination of our algorithm.

\begin{proposition}\label{prop:bigO}
  Let $I\subset\mathbb{R}$ be a compact interval and let $f$ analytic on
  an open set $\Omega \supset I$ and having no multiple zero.
  For $[a,b]\subset I$, let $h$ be the cubic
Hermite interpolant of $f$ at $a$ and $b$, and $Q$ from definition \ref{def:quality}.
Then
\[
    \sup_{x\in[a,b]} Q(x)=O((b-a)^{3}).
\]
\end{proposition}

\begin{proof}
  Set $E(x)=f(x)-h(x)$ and $L = b - a$, let
  $M = \|f^{(4)}\|_\infty$ on $I$. By fact \ref{fact:hermite}
  we have
\[
    |E(x)| \leqslant \frac{M}{384}L^4,
    \qquad
    |E'(x)| \leqslant \frac{2M}{81}L^3.
\]
Let us denote $\alpha=\frac{M}{384}L^4$ and $\beta=\frac{2M}{81}L^3$.
Since
\[
    \frac{|u-v|}{|u|+|v|} \leqslant \min\left(1, \frac{|u-v|}{|u|}\right),
\]
we have
\[
    Q(x) \leqslant \min\left(1,\frac{\alpha}{|f(x)|}\right)\min\left(1,\frac{\beta}{|f'(x)|}\right)
\]

The assumption that $f$ has no multiple zero implies that
$f(x)$ and $f'(x)$ never vanish simultaneously. By compactness, we define
\[
    \delta=\min_{x\in I}\bigl(|f(x)|+|f'(x)|\bigr)>0.
\]
We now assume that $L$ is small enough to have $\sqrt{\alpha} < \frac{\delta}{2}$
and $\sqrt{\beta} < \frac{\delta}{2}$
Hence, if $|f(x)|\leqslant\sqrt{\alpha}$, then $|f'(x)|\geqslant \delta - \sqrt{\alpha} \geqslant \frac{\delta}{2}$, and therefore
\[
    Q(x) \leqslant \frac{2\beta}{\delta} \in O(L^3).
\]
Similarly, if $|f'(x)| \leqslant \sqrt{\beta}$, then
\[
    Q(x) \leqslant \frac{2\alpha}{\delta} \in O(L^4).
\]
Otherwise $|f(x)| > \sqrt{\alpha}$ and $|f'(x)| > \sqrt{\beta}$,
\[
    Q(x) \leqslant \frac{\alpha\beta}{|f(x)|\,|f'(x)|}
     \leqslant \sqrt{\alpha\beta} \in O(L^{3.5})
\]
Hence, for $L$ sufficiently small to have $\sqrt{\alpha} < \frac{\delta}{2}$
and $\sqrt{\beta} < \frac{\delta}{2}$, we have $Q(x) \in O(L^{3})$.
\end{proof}

\begin{proposition}\label{prop:side}
Let $f$ be analytic on an open set $\Omega \subset \mathbb{R}$ having no multiple zero, let $I = [A,B] \subset \Omega$
and $a < b \in I$.
Let $h$ be the cubic Hermite interpolant of $f$ at $a$ and $b$
and let $\varepsilon\in]0,1[$.
There exists $\eta_\varepsilon > 0$ depending only
on $f$, $I$ and $\varepsilon$, such that
\[\text{if}\quad Q(x) > \varepsilon \quad\text{and}\quad h'(x)=0
\quad\text{then}\quad
(x-a)>\eta_\varepsilon.
\]
By symmetry, the same result holds near the endpoint $b$.
\end{proposition}

\begin{proof}
First, if $f$ is a polynomial of degree at most 3, then $Q$ is identically 0
as $f = h$ and the result is immediate. Otherwise, we denote $M =
\|f^{(4)}\|_\infty > 0$.
Since $f$ has no multiple zero and $I$ is compact,
we define, as in the previous proof,
\[
\delta=\min_{t\in I}\bigl(|f(t)|+|f'(t)|\bigr)>0.
\]
Using the last inequality in fact \ref{fact:hermite},
there exists a constant $C$, depending only on $f$ and $I$, such that
$|h''(t)|\leqslant C$ for any $t\in I$.
Consequently, if $h'(x)=0$,
\[
|f'(a)|
=|h'(a)-h'(x)|
\leqslant C(x-a).
\]
Choose $\eta = \frac{\delta}{2C}$. If $x-a<\eta$, then
$C(x-a) < \frac{\delta}{2}$ and
\[
|f(a)|
\geqslant \delta-|f'(a)|
\geqslant \frac{\delta}{2}.
\]
By continuity of $f$, after possibly decreasing $\eta$, using a uniform bound
on $I$ for $|f'|$ we also have
\[
|f(x)|\geqslant\frac{\delta}{4} \quad (\text{if } h'(x) = 0, x-a < \eta).
\]
The Hermite interpolation remainder from fact \ref{fact:hermite} gives
\[
|f(x)-h(x)| \leqslant \frac{M}{24} (x-a)^2(b-x)^2.
\]
Since the second factor in the definition of $Q$ is $1$ when $h'(x) = 0$,
\[
Q(x)
=
\frac{|f(x)-h(x)|}{|f(x)|+|h(x)|}
\leqslant
\frac{M}{6\delta}(x-a)^2(b-x)^2
\leqslant
\frac{M}{6\delta}
(x-a)^2(b-a)^2.
\]
Let us define
\[
\eta_\varepsilon = \min\left(\eta,\frac{1}{B-A}\sqrt{\frac{6\varepsilon\delta}{M}}\right)
\]
Thus, from $B - A \geqslant b - a$ and $(x - a) \leqslant \eta_\varepsilon$ we have
\[
Q(x) \leqslant \frac{M}{6\delta}
\eta_\varepsilon^2(b-a)^2\leqslant \varepsilon \frac{(b-a)^2 }{(B-A)^2} \leqslant \varepsilon.
\]
The result follows by contraposition.
The proof near $b$ is identical.
\end{proof}

\begin{definition}\label{def:good}
  Let $\varepsilon \in ]0,1[$ and $n \in \mathbb{N}$, we say that $[a,b]$ is $(\varepsilon,n)$-good
  If $Q(x) < \varepsilon$ for all the test points chosen as follows:
\begin{itemize}
\item  The first test points are the real roots of the derivative $h'$ that
  belong to $]a,b[$. These points correspond to the critical points of $h$
      and are therefore natural locations. Note for such points, as $h'(x) =
      0$, we have  $Q(x) = \frac{|f(x) - h(x)|}{|f(x)| + |h(x)|}$ and
      $Q(x) < 1$ implies that $h$ and $f$ have the same sign.

    \item The next evaluation points are $n$ points $x_1,\dots,x_n$ with
      $x_k = a + k \frac{b-a}{n+1}$ (hence $a,x_1,\dots,x_n,b$ are regularly spaced).
      Remark: we don't know if this choice is optimal.
\end{itemize}
\end{definition}

\begin{algorithm}[root isolation algorithm]
Our algorithm is a recursive function that considers an interval $[a,b]$.
It takes as fixed parameters a way to compute the value of a function $f$ and its
derivative, $\varepsilon \in ]0,1[$ and $n \in \mathbb{N}^*$.
      It returns a list of sub-intervals of $[a,b]$ and processes as follows:
\begin{itemize}
  \item
    If the interval $[a,b]$ under consideration is not $(\varepsilon,n)-good$ we subdivide
    the interval at a tested point that maximizes $Q$.
  \item Otherwise we use the sign
    changes of $h$ (we computed the roots of $h'$ to test for goodness) to collect at most 3 intervals with exactly one root of $h$
    and at least one root of $f$. The lower bound on the number of roots of
    $f$ follows from the fact that $h$ and $f$ have the same sign in $a$,
    $b$ and the roots of $h'$ (see the remark in the first item of definition \ref{def:good}).
\end{itemize}
\end{algorithm}

Remark: We have not tested approaches using Newton step to improve performance in the
case of clustered roots like for Mignotte polynomials as in \cite{KAR2016}.

\begin{proposition}
  Our root isolation algorithm terminates for an analytic function $f$ on an
  open set $\Omega \subset \mathbb{R}$, with no multiple roots and
  starting from any interval $I = [A,B] \subset \Omega$. It gives a list of intervals in $I$ that
  contain at least one root of $f$ each. This is true for any value of $\varepsilon \in
    ]0,1[$ and $n \in \mathbb{N}^*$.
\end{proposition}

\begin{proof}
  Termination is a consequence of proposition \ref{prop:bigO} and
  \ref{prop:side}. The former ensures that small enough interval are
  $(\varepsilon,n)-good$. The latter ensure that intervals decrease to a $0$ width.
  Indeed, if the subdivision point is one of the regularly spaced test points, this is
  immediate. If the subdivision point $x$ is a root of $h'$, the derivative of the
  interpolant of $f$ on the considered interval $[a,b]$, the latter lemma
  ensures that $x - a > \eta_\varepsilon$ and $b - x  > \eta_\varepsilon$ Hence
  $x - a < (b - a) - \eta_\varepsilon$ and $b - x < (b - a) - \eta_\varepsilon$.

  The fact
  that each interval contains at least one root of $f$ is explained in the
  description of the algorithm.
\end{proof}

\begin{conjecture}
  For any analytic function with no multiple
  roots, we can find $\varepsilon \in ]0,1[$ and $n \in \mathbb{N}^*$ such that the above algorithm misses
      no roots (never or with very low probability, we do not know if we can get
      an exact algorithm). Moreover, our experiments seem to show that $\varepsilon$ and $n$ may
      be relatively small and independent of the degree and the size of
      interval for many families of polynomials.
\end{conjecture}

\section{Experimental Results}

\subsection{Software availability.}
All the code used in this work is publicly available. The pure root isolation
algorithm is implemented in Julia and can be used independently for reproducing
the numerical experiments. The implicit surface visualization experiments rely
on the CamlSurf framework, available at
\url{https://github.com/craff/CamlSurf}. In CamlSurf, the root finding
algorithm has been rewritten in the GLSL language.

\subsection{Tests on classical polynomials}

The algorithm has been tested on several classical families of
polynomials, including Chebyshev polynomials, Wilkinson polynomials,
Mignotte polynomials and randomly generated examples.

We give below the definition used to compute the value of $P$ which we also
differentiate to compute the value of $P'$. We also give the interval in
which we search roots which is often larger than the obvious one. As we work
with arbitrary function, we do not use classical bounds for polynomial roots.

For most of these tests, we use the following values for the constants:
\[\varepsilon = 0.2 \hspace{1cm} n = 3\]. We indicate when we need larger values to
avoid missing roots.

Far from the roots, our algorithm often needs to many subdivisions because a
polynomial of degree 3 cannot approximate well a higher-degree polynomials
far from the origin. An idea we tested is composing with the function $\Phi_\lambda(x)
= \lambda \sinh^{-1}(\frac{x}{\lambda})$. In works well in some cases, not in others
(notably Wilkinson's Polynomials)

For our algorithm, we give three curves:
\begin{itemize}
\item Number of subdivisions
\item Time to isolate the roots
\item Time to isolate the roots and get the maximum possible precision using
  dichotomie (we stop when there is no more floating point value in the
  interval).
\end{itemize}

For polynomials (i.e. not when composing with $\Phi_\lambda$), we compare our algorithm
with two widely available implementations of root isolation : pari/gp and RS
(available in maple or Julia). The library RS follows Rouillier and Zimmermann work
\cite{ROUILLIER200433} using Descartes' rule while Pari uses a modification of
Uspensky's method follwing the same work of Rouillier and Zimmermann. We
consistently outperform both, sometimes by a large margin. This can be
attributed to the fact that our algorithm can compute with much smaller
precision (i.e. 64-bit float), and the fact that for some computation, we do
not develop the polynomial. We even use the trigonometric formulation for
Chebyshev Polynomials.

\paragraph{Legendre polynomials}

\[P_0 = 1, P_1 = X, P_n = (2n - 1) P_{n-1} - (n-1) P_{n-2} \hbox{ search
  on } [-10,10]\].

We get a linear number of subdivisions and a quadratic computing time because the
evaluation of this polynomial is in $O(n)$.

\begin{figure}
  \centering
  \includegraphics[width=0.47\textwidth]{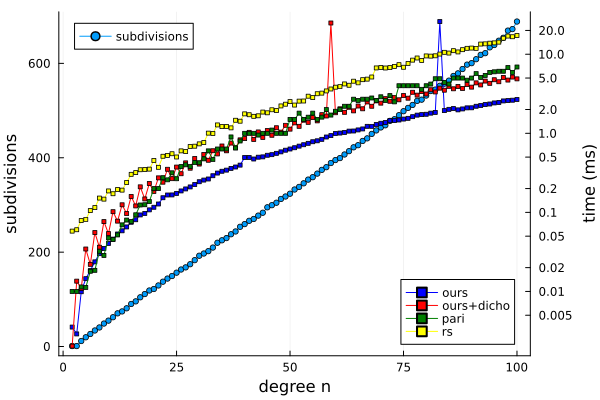}
  \includegraphics[width=0.47\textwidth]{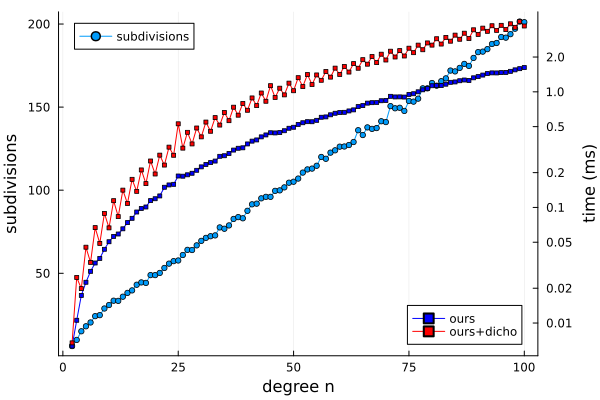}
  \caption{Experimental results for Legendre polynomials (composed with
    $\Phi_{10}$ on the right).\\ left: $\varepsilon=0.2,n=3$, right:  $\varepsilon=0.1,n=6$.}
  \label{fig:legendre}
\end{figure}

Moreover, usual algorithm, implemented with exact arithmetic, using Descartes'
rule of sign or Sturm sequences cannot reach such high degree without error
when they are implemented using only computation with 64-bit floating-point
arithmetic. Errors start to occur in general much sooner.

\paragraph{Chebyshev polynomials}

\[P_n = \cos(n \cos^{-1}(X)) \hbox{ search on } [-10,10]\]

Here we have a linear number of subdivisions and computing time.

\begin{figure}
  \centering
  \includegraphics[width=0.47\textwidth]{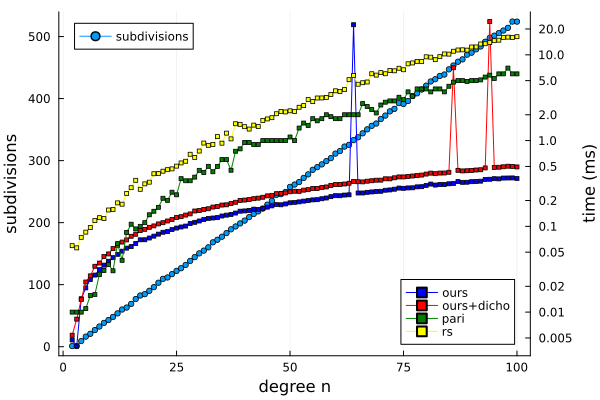}
  \includegraphics[width=0.47\textwidth]{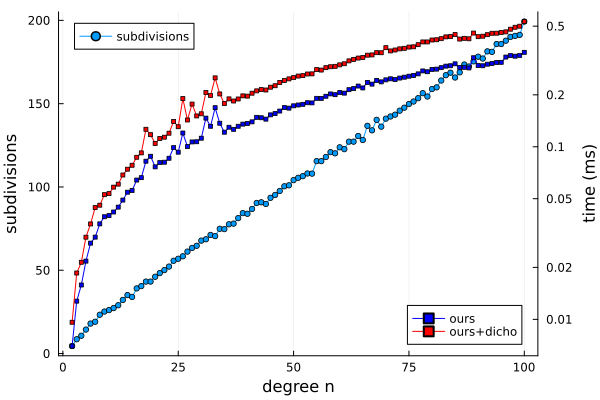}
  \caption{Experimental results for Chebyshev polynomials (composed with
    $\Phi_{10}$ on the right).\\ left: $\varepsilon=0.2,n=3$, right:  $\varepsilon=0.1,n=6$.}
  \label{fig:chebyshev}
\end{figure}

\paragraph{Wilkinson polynomials}

\[P_n = \Pi_{k=1}^n (x - k)  \hbox{ search on } [-n,2n]\]

\begin{figure}
  \centering
  \includegraphics[width=0.47\textwidth]{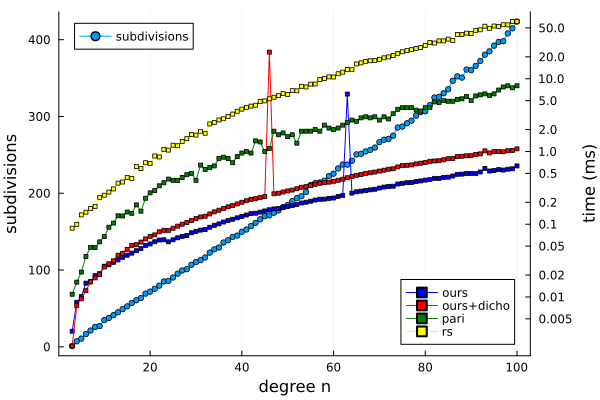}
  \includegraphics[width=0.47\textwidth]{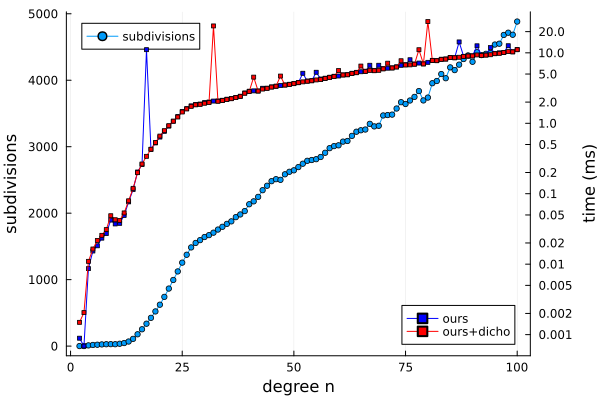}
  \caption{Experimental results for Wilkinson polynomials (composed with
    $\Phi_{10^{20}}$ on the right).\\ left: $\varepsilon=0.2,n=3$, right:  $\varepsilon=0.015,n=3$.}
  \label{fig:wilkinson}
\end{figure}

Wilkinson polynomials require more subdivisions, hence more time than the two
previous family. Still, it seems that the number of subdivisions is linear.
Moreover, the composition with $\Phi_\lambda$ does not work even for large value of
$\lambda$: it does not diminish the number of subdivisions and requires much smaller
$\varepsilon$. Still, running times are reasonable!

\paragraph{Geometric polynomials}

\[P_n = \Pi_{k=0}^{n-1} (x - 2^{-k})  \hbox{ search on } [-10,10]\]

\begin{figure}
  \centering
  \includegraphics[width=0.47\textwidth]{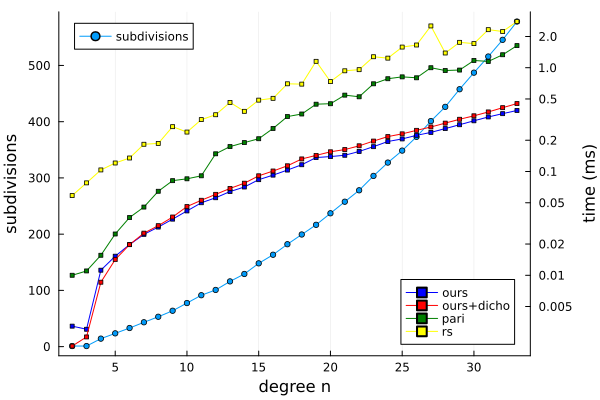}
  \includegraphics[width=0.47\textwidth]{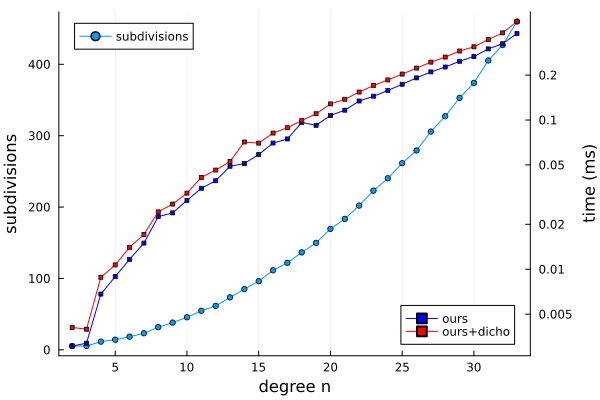}
  \caption{Experimental results for geometric polynomials (composed with
    $\Phi_{10}$ on the right).\\ left: $\varepsilon=0.2,n=3$, right:  $\varepsilon=0.15,n=3$.}
  \label{fig:geometric}
\end{figure}

These polynomials have a large cluster of roots near 0 and need more than a
linear number of subdivisions. Still the number of subdivisions is relatively
small.

\paragraph{Mignotte polynomials}

\[P_{n,p} = X^n - 2(2^p X - 1)^2 \hbox{ search on } [-M,M] \hbox{ with } M = \max\left(2,\frac{2^{2p+3}}{n-2}\right)\]

\begin{figure}
  \centering
  \includegraphics[width=0.47\textwidth]{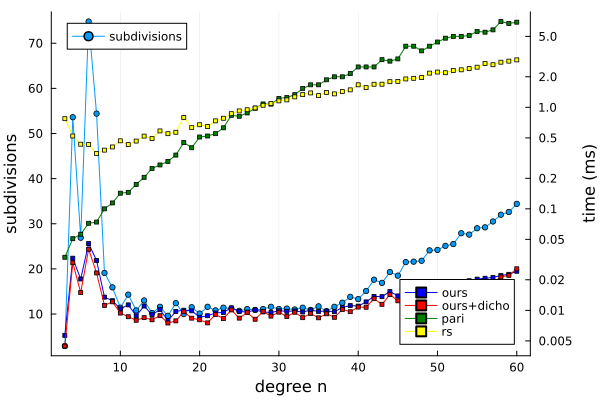}
  \includegraphics[width=0.47\textwidth]{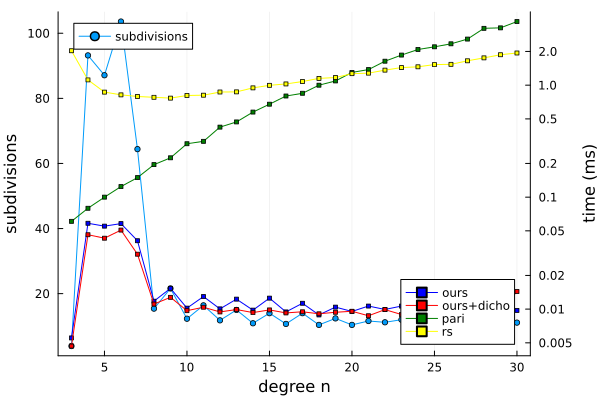}
  \caption{Experimental results for Mignotte $P_{n,16}$ and $P_{n,32}$
    polynomials.\\ left: $\varepsilon=0.15,n=3$, right:  $\varepsilon=0.15, n=3$.}
  \label{fig:mignotte}
\end{figure}

These polynomials have 3 or 4 roots including two very near
0. Again, double precision limits the degree we can reach: around 60 for
$P_{n,16}$ and 30 for $P_{n,32}$. Still the behaviour of the algorithm is very
good : a relatively constant number of subdivisions!

There is a surprising anomaly for low degree, because of an early subdivision near a root
of $f$. This is probably related to a very small $\eta_\varepsilon$ constant from proposition
\ref{prop:side} and a subdivision point to near a root of $f$.

\paragraph{Random polynomials}
Random test polynomials we generated in two ways:

First, uniformly for Bombieri-Weyl norm (leading to an average of $O(\sqrt{n})$
roots in degree $n$, left graphic) \cite{EDELMAN1995}
and choosing their roots following a normal
distribution $\mathcal{N}(0,2)$, therefore presenting a root cluster near $0$.

We use the interval $[\min(-10,3m),\max(10,3M)]$ where $m$ and $M$ are the
smallest and largest root of the polynomial in both cases.
\begin{figure}
  \centering
  \includegraphics[width=0.47\textwidth]{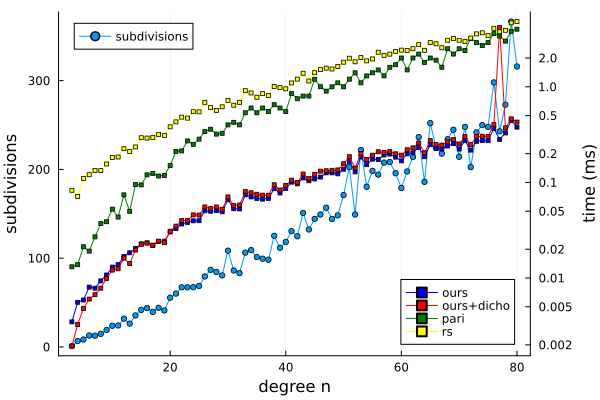}
  \includegraphics[width=0.47\textwidth]{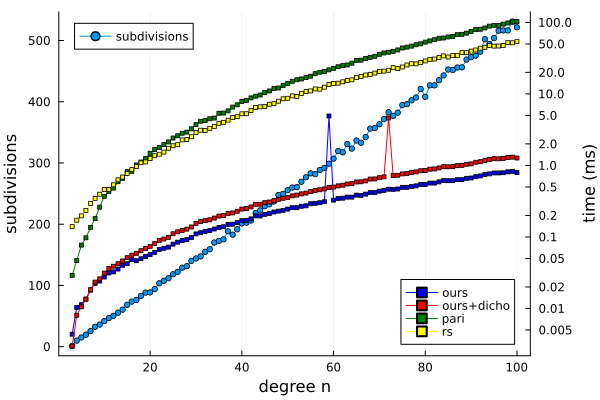}
  \caption{Experimental results for random polynomials (left: uniformly for
    Bombieri norm, right: random roots).\\ left: $\varepsilon=0.15,n=3$, right:  $\varepsilon=0.1,n=3$.}
  \label{fig:random}
\end{figure}

Remark: polynomials with random roots require very large coefficients after
expansion. This is why our method is faster in that case.

\subsection{Experiments with visualization of implicit surfaces}

For these experiments to have (almost) no visible bad pixel, we use slightly
smaller constant, the same for all examples:
\[\varepsilon = 0.15 \hspace{1cm} n = 3\].

\begin{figure}
  \centering
  \includegraphics[width=0.6\textwidth]{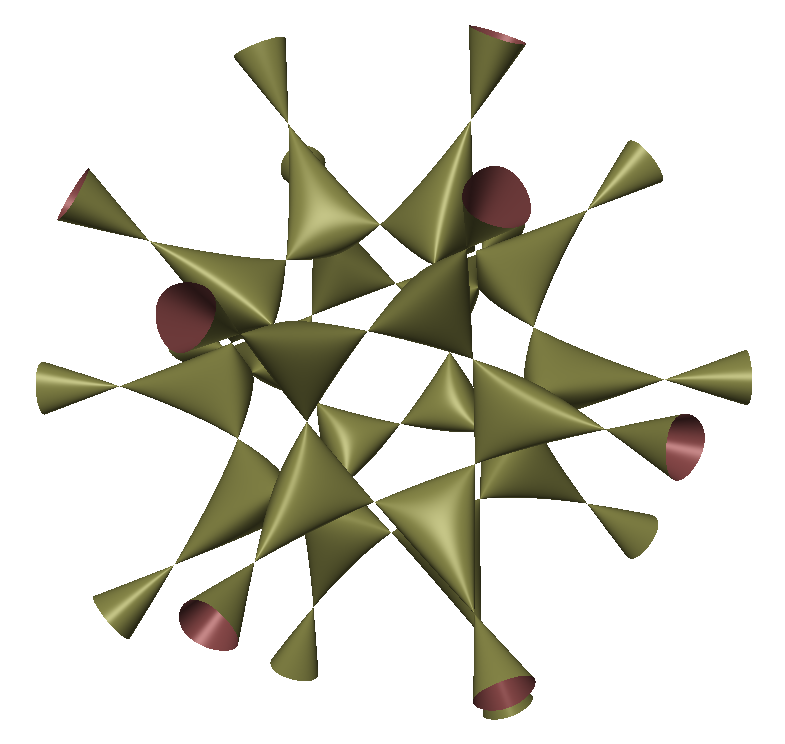}
  \caption{Barth sextic, 100 FPS.}
  \label{fig:barth}
\end{figure}

The proposed interpolation criterion is also used in an implicit surface
visualization framework. The algorithm is integrated into a GPU-based renderer
using single precision floating-point arithmetic (IEEE 754 float32).

In most of our examples, we give the most complex here, the rendering rate
remains above 30 FPS (frames per second), allowing
interactive exploration and animation of the surfaces.
One exception is the
heptic surface from O. Labs, whose very large algebraic expression leads to a
substantial increase in evaluation cost at only 20 FPS. We also show two non
polynomial functions. Their high number of roots and the use of
transcendental functions also require more time.

\begin{figure}
  \centering
  \includegraphics[width=0.6\textwidth]{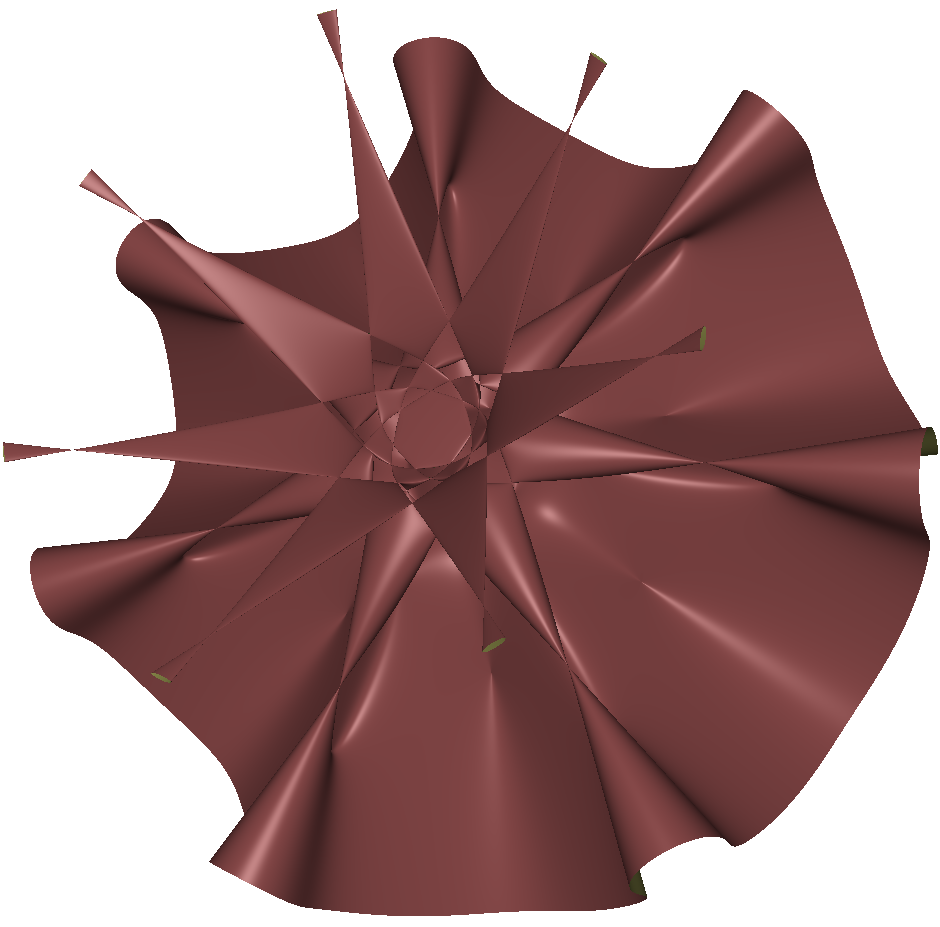}
  \caption{Labs heptic, 20 FPS.}
  \label{fig:labs7}
\end{figure}

Despite the use of single precision arithmetic and the absence of exact
certification, no visible numerical artifact or incorrect root
classification was observed in our experiments. These results suggest that the
local interpolation test remains highly reliable even in a demanding real-time
geometric setting.

Game: there are ~4 bad pixels in the figure \ref{fig:trigo1} that was produced
with an early version of the algorithm. The latest implementation following
exactly the algorithm in this paper does not feature any bad pixel with the
above mentioned value of $\varepsilon = 0.15$ and $n = 3$.

The experiments were conducted using only an integrated graphics processor,
an Intel UHD Graphics based on the Alder Lake-P GT1 architecture, rather than a
dedicated GPU. Even with this relatively limited hardware, the proposed
approach achieves interactive frame rates in most cases.

\begin{figure}
  \centering
  \includegraphics[width=0.6\textwidth]{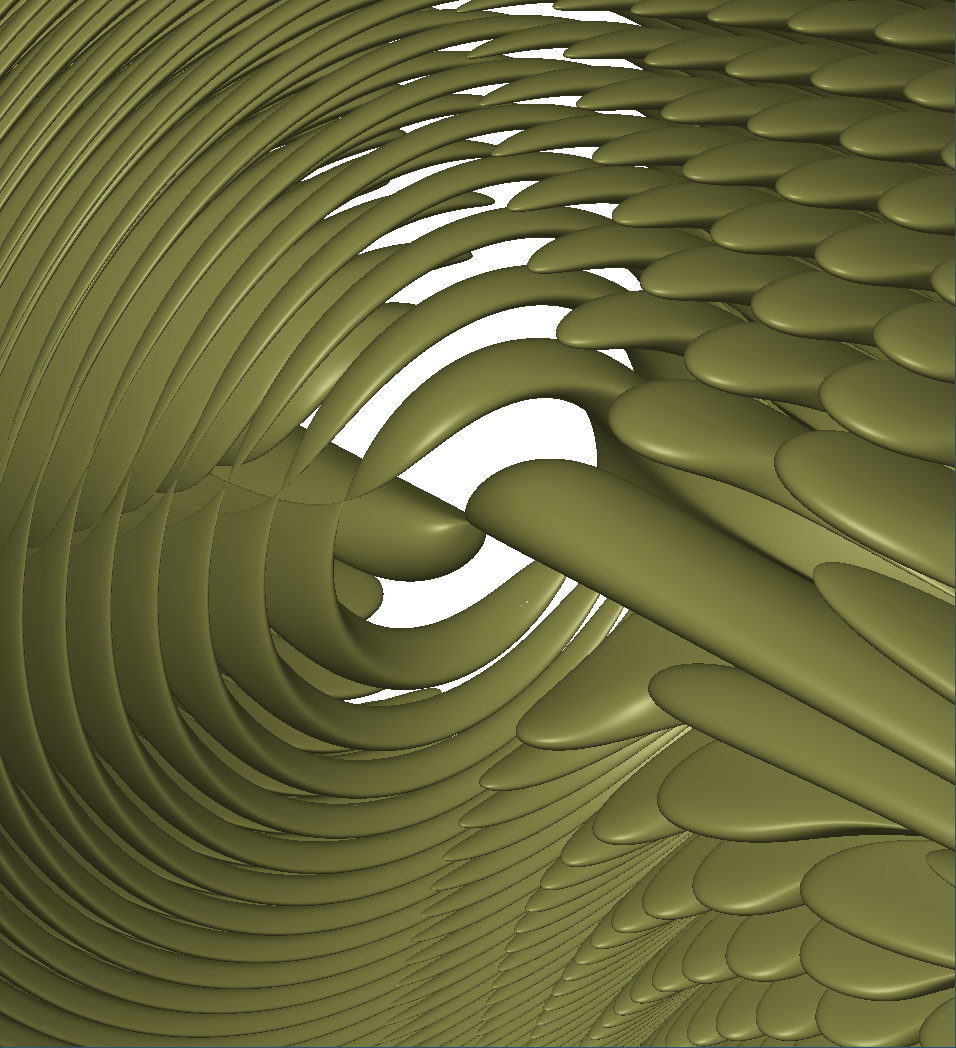}
  \caption{$\cos(y^2+z^2)-\sin(xy)-e^{\frac{xz}{2}} = 0$, 30 FPS.}
  \label{fig:trigo1}
\end{figure}

\begin{figure}
  \centering
  \includegraphics[width=0.6\textwidth]{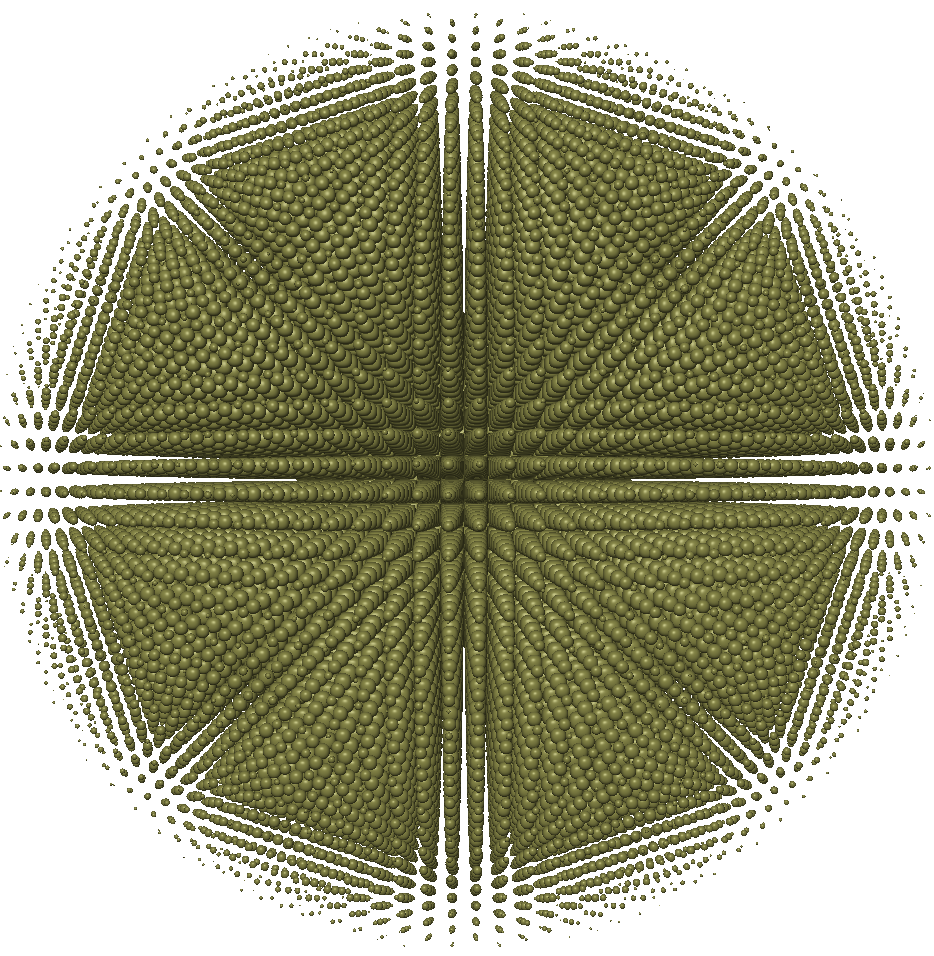}
  \caption{$\sin(20 x)\sin(20 y)\sin(20 z) - \frac{1}{10} (x^2+y^2+z^2-1) = 0$, 10 FPS.}
  \label{fig:trigo2}
\end{figure}

\bibliography{biblio}
\bibliographystyle{plain}

\end{document}